\documentclass[sigconf]{acmart}

\renewcommand\footnotetextcopyrightpermission[1]{} 
\usepackage{booktabs} 

\setcopyright{rightsretained}

\usepackage{enumitem}
\usepackage{tikz}

\usepackage{bbold, amsmath}
\usepackage{stmaryrd}
\usepackage{mathtools}
\DeclarePairedDelimiter\ceil{\lceil}{\rceil}
\DeclarePairedDelimiter\floor{\lfloor}{\rfloor}

\DeclareMathOperator*{\argmax}{arg\,max}

\begin{document}

\title{A Dynamic Game Analysis and Design of Infrastructure Network Protection and Recovery}

\author{Juntao Chen}
\affiliation{%
  \institution{ECE Department,
 New York University, Brooklyn NY 11201}
}
\email{jc6412@nyu.edu}

\author{Corinne Touati}
\affiliation{%
  \institution{INRIA, F38330 Montbonnot Saint-Martin, France}
}
\email{corinne.touati@inria.fr}

\author{Quanyan Zhu}
\affiliation{%
  \institution{ECE Department,
 New York University, Brooklyn NY 11201}
}
\email{ qz494@nyu.edu}

\begin{abstract}

Infrastructure networks are vulnerable to both cyber and physical attacks. Building a secure and resilient networked system is essential for providing reliable and dependable services. To this end, we establish a two-player three-stage game framework to capture the dynamics in the infrastructure protection and recovery phases. Specifically, the goal of the infrastructure network designer is to keep the network connected before and after the attack, while the adversary aims to disconnect the network by compromising a set of links. With costs for creating and removing links, the two players aim to maximize their utilities while minimizing the costs. In this paper, we use the concept of subgame perfect equilibrium (SPE) to characterize the optimal strategies of the network defender and attacker. We derive the SPE explicitly in terms of system parameters. Finally, we use a case study of UAV-enabled communication networks for disaster recovery to corroborate the obtained analytical results.
\end{abstract}

\maketitle

\section{Introduction}
Infrastructure networks are increasingly connected due to the integration of the information and communications technologies (ICTs). For example, the introduction of smart meters has enabled the communications between the users and the utility companies. Communications with roadside units in vehicular networks can provide safety warnings and traffic information.

However, infrastructure networks are vulnerable to not only physical attacks (e.g., terrorism, theft or vandalisms) but also cyber attacks. These attacks can damage the connectivity of the infrastructure system and thus results in the performance degradation and operational dysfunction.  For instance, an adversary can attack the road sensor units and create traffic congestion \cite{al2012survey}. As a result, the transportation system can break down due to the loss of roads. An adversary can also launch denial-of-service attacks to disconnect communication networks \cite{pelechrinis2011denial}, resulting in inaccessibility of relevant database for air travel or financial transactions.

The cyber-physical nature of the infrastructure can also enable the coordinated attacks on the infrastructure systems that allow an adversary to use both cyber and physical approaches to disconnect networks. Therefore, infrastructure protection plays a significant role to maintain the connectivity of the infrastructure networks. One way to protect the network is to create redundant links in the network so that networks can be still connected despite arbitrary removal of links.
This approach has been used in traffic networks by creating multiple modes of transportation, in communication networks by adding extra wired or wireless links, and in supply chain networks by making orders from multiple suppliers.

Adding link redundancy is an effective approach when there is no knowledge of the target of the attacker, and thus the objective of the network designer is to secure the network by making the network robust to arbitrary removal of a fixed number of links.
However, it becomes expensive and sometimes prohibitive when the cost for creating links is costly, and the attacker is powerful.
Therefore, a paradigm shift to emphasize the recovery and response to attacks is critical, and the infrastructure resilience becomes essential for developing post-attack mechanisms to mitigate the impacts. 
With a limited budget of resources, it is essential to develop an optimal post-attack healing mechanism as well as a pre-attack secure mechanism holistically and understand the fundamental tradeoffs between security and resilience in the infrastructures.

To this end, we establish a two-player dynamic three-stage network game formation problem in which the infrastructure network designer aims to keep the network connected before and after the attack, while the objective of the adversary is to keep the network disconnected after the attack. Note that each player has a cost on creating or removing links.
Specifically, at the first stage of the game, the infrastructure network designer first creates a network with necessary redundancies by anticipating the impact of adversarial behavior. Then, an adversary attacks at the second stage by removing a minimum number of links of the network. At the last stage of the game, the network designer can recover the network after the attack by adding extra links to the attacked network. 

The resilience of the network is characterized by the capability of the network to maintain connectivity after the attack and the time it takes to heal the network. The security of the infrastructure is characterized by the capability of the network to withstand the attack before healing. 
Adding a large number of redundancies to the network can prevent the attack from disconnecting the network, but this approach can be costly. Hence, it is important to make strategic decisions and planning to yield a protection and recovery mechanism for the infrastructure with a minimum cost.

We adopt subgame perfect Nash equilibrium (SPE) as the solution concept of the dynamic game. We observe that with sufficient capabilities of recovery, the infrastructure can mitigate the threats by reducing the incentives of the attackers. We analyze SPE of the game by investigating two different parameter regimes. Further, we develop an optimal post-attack network healing strategy to recover the infrastructure network.  When an attacker is powerful (attack cost is low), we observe that the defender needs to allocate more resources in securing the network to reduce the incentives of the attacker. In addition, agile resilience and fast response to attacks are critical in mitigating the cyber threats in the infrastructures. 

\textit{\textbf{Related Works:}} Security is a critical concern for infrastructure networks \cite{ten2010cybersecurity, lewis2014critical}. The method in our work is relevant to the recent advances in adversarial networks \cite{dziubinski2013network,bravard2017optimal} and network formation games \cite{chen2016resilient,chen2016interdependent}. In particular, we jointly design the optimal protection and recovery strategies for infrastructure networks.

\textit{\textbf{Organization of the Paper:}} The rest of the paper is organized as follows. Section \ref{formulation} formulates the problem. Dynamic game analysis are presented in Section \ref{analysis}. Section \ref{SPE_analysis} derives the SPE of the dynamic game. Case studies are given in Section \ref{case_study}, and Section \ref{conclusion} concludes the paper. 

\section{Dynamic Game Formulation}\label{formulation}

In this section, we consider an infrastructure system represented by a set ${\mathcal N}$ of $n$ nodes. The infrastructure designer can design a network with redundant links before the attack for protection and adding new links after the attack for recovery. The sequence of the actions taken by the designer and the attacker is described as follows:
\begin{itemize}
\item A Designer ($D$) aims to create a network between these nodes and protect it against a malicious attack;
\item After some time of operation, an Adversary ($A$) puts an attack on the network by removing a subset of its links;
\item Once the $D$ realizes that an attack has been conducted, it has the opportunity to heal its network by constructing new links (or reconstructing some destroyed ones).
\end{itemize}
In addition, the timing of the actions also play a significant role in determining the optimal strategies of both players. We normalize the horizon of the event from the start of the preparation of infrastructure protection to a time point of interest as the time internal $[0, 1]$. This normalization is motivated by the observation made in \cite{DOEReport} where the consequences of  fifteen major storms occurring between 2004 and 2012 are plotted over a normalized duration of the event. We let $\tau$ and $\tau_R$ represent, respectively, the fraction of time spent before the attack (system is fully operational) and between the attack and the healing phase. This is illustrated in Figure~\ref{fig:time}.

\begin{figure}[h]
\centering{
\begin{tikzpicture}[scale=0.8]
\draw (0,0) -- (10,0);
\draw (0.5, 0.5) -- (0.5,-0.5) node [anchor=north]{$0$};
\draw (5,0.5) -- (5,-0.5) node [anchor=north] (t) {$\tau$};
\draw (7,0.5) -- (7,-0.5) node [anchor=north] (tr) {$\tau+\tau_R$}; 
\draw (9.5,0.5) -- (9.5,-0.5) node [anchor=north]{$1$};
\node (A) at (5,-2) {Attack};
\node (R) at (7,-2) {Recovery};
\path[->] (A) edge node {} (t);
\path[->] (R) edge node {} (tr);
\end{tikzpicture}
}
\caption{Attack and Defense Time Fractions.}
\label{fig:time}
\end{figure}
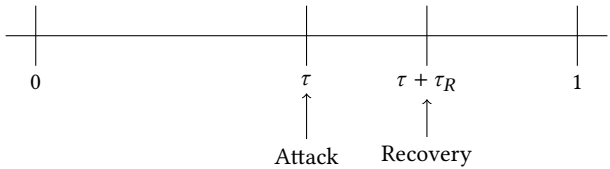
The goal of the designer or the defender is to create protection and recovery mechanisms to keep its network operational, i.e., connected in this case. Let ${\mathcal E}_1$ be the set of links created by the defender initially (i.e. at time $0$). ${\mathcal E}_A \subseteq {\mathcal E}_1$ is the set of links removed (attacked) by the adversary and $\mathcal{E}_2$ is the set of links created by the defender after the attack (at fraction $\tau+\tau_R$ of the time horizon). Regardless of the time stamp, creating (resp. removing) links has a unitary cost $c_D$ (resp. $c_A$). The adversary aims to disconnect the network. Thus, for any set ${\mathcal E}$, we define $\mathbb{1}_E$ which equals $1$ if the graph $({\mathcal N},{\mathcal E})$ is connected and $0$ otherwise. Values $\tau$, $\tau_R$, $c_A$ and $c_D$ are common knowledge to both the Designer and the Adversary. As a tie-breaker rule, assume that if the output is the same for the Adversary, the Adversary chooses to attack the network with the strongest number of link removals. Similarly, the Designer chooses not to create links if its utility is the same.

Therefore, the utility for the designer (resp. adversary) is equal to the fraction of time the network is connected (resp. disconnected) minus the costs of creating (resp. removing) the links. Hence the payoff functions of the designer and the adversary are represented by $U_D$ and $U_A$, respectively, as follows:
$$
\begin{array}{@{}l@{}l@{}}
U_D({\mathcal E}_1, {\mathcal E}_2, {\mathcal E}_A) =  &(1-\tau-\tau_R) \mathbb{1}_{E_1\backslash E_A\cup E_2} +  \tau \mathbb{1}_{E_1} \\
 & \hfill + \tau_R \mathbb{1}_{E_1\backslash E_A} -c_D (|{\mathcal E}_1| + |{\mathcal E}_2|), \\
U_A({\mathcal E}_1, {\mathcal E}_2, {\mathcal E}_A) = & (1-\tau-\tau_R) (1-\mathbb{1}_{E_1\backslash E_A\cup E_2}) -c_A |{\mathcal E}_A| \\
 &\hfill +  \tau (1-\mathbb{1}_{E_1})+ \tau_R (1-\mathbb{1}_{E_1\backslash E_A}), 
\end{array}
$$
where $|\cdot|$ denotes the cardinality of a set.

As the players are strategic, we study the SPE and analyze the strategies of the players to the sets $({\mathcal E}_1, {\mathcal E}_A, {\mathcal E}_2)$. Thus, we seek triplets $({\mathcal E}_1, {\mathcal E}_A, {\mathcal E}_2)$ such that ${\mathcal E}_2$ is a best response to $({\mathcal E}_1,{\mathcal E}_A)$ and that given ${\mathcal E}_1$, $({\mathcal E}_A, {\mathcal E}_2)$ is also a SPE. In other words, the SPE involves the analysis of the following three sequentially nested problems starting from the last stage of the designer's recovery problem to the first stage of the designer's protection problem:
\begin{itemize}
\item[(i)] Given the strategies ${\mathcal E}_1$ and ${\mathcal E}_A$, player $D$ chooses\\ ${\mathcal E}_2^*({\mathcal E}_1, {\mathcal E}_A) \in \argmax_{{\mathcal E}_2} U_D ({\mathcal E}_1, {\mathcal E}_A, {\mathcal E}_2)$;
\item[(ii)] Given ${\mathcal E}_1$, the adversary chooses\\ ${\mathcal E}_A^*({\mathcal E}_1) \in \argmax_{{\mathcal E}_A} U_A ({\mathcal E}_1, {\mathcal E}_A, {\mathcal E}_2^*({\mathcal E}_1, {\mathcal E}_A))$;
\item[(iii)] Player $D$ chooses\\ ${\mathcal E}_1^* \in  \argmax_{{\mathcal E}_1}U_D ({\mathcal E}_1, {\mathcal E}_A^*({\mathcal E}_1), {\mathcal E}_2^*({\mathcal E}_1, {\mathcal E}_A))$.
\end{itemize}
The equilibrium solution $({\mathcal E}_1, {\mathcal E}_A, {\mathcal E}_2)$ that solves the above three problems consistently is an SPE of the two-player dynamic game. 

\section{Game Analysis}\label{analysis}
In this section, we analyze the possible configurations of the infrastructure network at SPE.

We first note that $c_A$ should be not too large, since otherwise $A$ cannot be a threat to $D$. Similarly, $c_D$ should be sufficiently small so that the $D$ can create a connected network:
\begin{lemma}
If $c_A > 1- \tau$, then $A$ has no incentive to attack any link.
If $c_D > \frac{1}{n-1}$, then $D$ has no incentive to create a connected network.
\end{lemma}
\begin{proof}
Suppose that $c_A > 1- \tau$. Let ${\mathcal E}_1$ be given and $B := \tau (1-\mathbb{1}_{E_1})$. If $A$ decides not to remove any link, then its payoff is $B + \tau_R (1-\mathbb{1}_{E_1})+(1-\tau-\tau_R) (1-\mathbb{1}_{E_1\cup E_2}) \geq B$. Otherwise, $|{\mathcal E}_A| \geq 1$ and 
$U_A({\mathcal E}_1, {\mathcal E}_2, {\mathcal E}_A) \leq  B + (1-\tau-\tau_R) (1-\mathbb{1}_{E_1\backslash E_A\cup E_2}) -c_A
+ \tau_R (1-\mathbb{1}_{E_1\backslash E_A}) \leq B + 1-\tau-c_A < B$. 
Thus, it is a best response for $A$ to play ${\mathcal E}_A = \emptyset$.
Similarly, if $c_D > \frac{1}{n-1}$, then if $D$ plays ${\mathcal E}_1 = {\mathcal E}_2 = \emptyset$, its utility is $0$. Otherwise, its utility is bounded above by $1-(n-1) c_D$ which corresponds to a connected tree network with the minimum number of links. 
\end{proof}

In the following, we thus suppose that $c_A < 1-\tau$ and $c_D < \frac{1}{n-1}$.


Note that the SPE can correspond only to a set of situations:
\begin{lemma}\label{lem:situations}
Suppose that $({\mathcal E}_1, {\mathcal E}_A, {\mathcal E}_2)$ is an SPE. Then, we are necessarily in one of the situations given in Table~\ref{tab}.
\begin{table}[h]$$\begin{array}{c|ccc}
\mathrm{Situation} & \mathbb{1}_{E_1} & \mathbb{1}_{E_1\backslash E_A} & \mathbb{1}_{E_1\backslash E_A\cup E_2}\\
\hline
1 & 1 & 1 & 1\\
2 & 1 & 0 & 1\\
3 & 1 & 0 & 0\\
4 & 0 & 0 & 1\\
5 & 0 & 0 & 0
\end{array}$$
\caption{The different potential combinations of values of $\mathbb{1}_{E_1}$, $\mathbb{1}_{E_1\backslash E_A}$ and $\mathbb{1}_{E_1\backslash E_A\cup E_2}$ at the SPE.}\label{tab}\vspace{-5mm}
\end{table}
\end{lemma}

\begin{proof}
Altogether, $8$ situations should be possible. However, if $ \mathbb{1}_{E_1} = 0$, then it is impossible that $\mathbb{1}_{E_1\backslash E_A} = 1$. Therefore, the situations where $(\mathbb{1}_{E_1}, \mathbb{1}_{E_1\backslash E_A}, \mathbb{1}_{E_1\backslash E_A\cup E_2})$ equaling to $(0,1,0)$ and $(0,1,1)$ are impossible. 

Further, if $\mathbb{1}_{E_1\backslash E_A}=1$, then it is impossible that $\mathbb{1}_{E_1\backslash E_A\cup E_2}=0$. Thus, the situation $(\mathbb{1}_{E_1}, \mathbb{1}_{E_1\backslash E_A}, \mathbb{1}_{E_1\backslash E_A\cup E_2})=(1,1,0)$ is impossible. All other combinations are summarized in Table~\ref{tab}.
\end{proof}

The shape of the SPE depends on the values of the parameters of the game. In particular, it depends on whether the $D$ has incentive to fully reconstruct (heal) the system after an attack of the $A$. More precisely, if $1-\tau-\tau_R > (n-1)c_D$, then the $D$ prefers to heal the network even if all links have been compromised by the attacker. Otherwise, there should be a minimum number of links remained after the attack for the $D$ to heal the network at the SPE. We analyze these two cases in Sections~\ref{sec:1} and \ref{sec:2}, respectively.

\section{SPE Analysis of the Game}\label{SPE_analysis}
Depending on the parameters, we derive SPE of the game in two regimes: $1-\tau-\tau_R > (n-1)c_D$ and the otherwise.

\subsection{Regime 1: $1-\tau-\tau_R > (n-1)c_D$}
\label{sec:1}

In the case where $1-\tau-\tau_R > (n-1)c_D$, the network always recovers to be connected after the attack. The potential SPE can occur in only three of the Situations in Table~\ref{tab}. More precisely:
\begin{proposition}\label{prop:1}
Suppose that $1-\tau-\tau_R > (n-1)c_D$ and let $k_A^R := \floor*{\frac{\tau_R}{c_A}}$. Then, the SPE of the game is unique and satisfies:
\begin{itemize}
\item If $\tau_R < c_A$, then $U_D = 1-(n-1)c_D$ and $U_A = 0$ (Situation $1$).
\item Otherwise, 
\begin{itemize}
\item if $\tau > c_D$ and $\tau_R > c_D \ceil*{\frac{n(k_A^R-1)}{2}}$ or if $\tau<c_D$ and $ \tau+\tau_R > c_D \ceil*{\frac{n(k_A^R-1)}{2}+1}$, then the SPE 
satisfies \\$\left\{\begin{array}{l}
U_D = 1-c_D \ceil*{\frac{n(k_A^R+1)}{2}} \\
U_A=0
\end{array}\right.$ (Situation $1$).
\item If $\tau > c_D$ and $\tau_R < c_D \ceil*{\frac{n(k_A^R-1)}{2}}$,  then the SPE satisfies $\left\{\begin{array}{l}
U_D = 1-\tau_R-n c_D \\ U_A = \tau_R-c_A \end{array}\right.$ (Situation $2$).
\item If $\tau < c_D$ and $ \tau+\tau_R < c_D \ceil*{\frac{n(k_A^R-1)}{2}+1}$, then the SPE 
satisfies\\
$\left\{\begin{array}{l}
U_D = 1-\tau-\tau_R-(n-1)c_D\\
U_A=\tau+\tau_R 
\end{array}\right.$ (Situation $4$).
\end{itemize}
\end{itemize}
\end{proposition}

Proposition \ref{prop:1} is a direct consequence of the following lemma:
\begin{lemma}\label{lem:case1}
Suppose that $1-\tau-\tau_R \geq (n-1)c_D$. The potential SPEs have the properties given in Table~\ref{tab:SPE}.
\begin{center}
\begin{table}[h]
{\small
$$\hspace{-1em}\begin{array}{|@{}c@{\,}|@{}c@{\;}c@{\;}c@{\,}|@{\,}c@{\;\;}c@{}|}
\mathrm{Situation} & |{\mathcal E}_1| & |{\mathcal E}_A| & |{\mathcal E}_2| & U_D & U_A \\
\hline
1 \& k_A^R>0 & \ceil*{\frac{n(k_A^R+1)}{2}} & 0 & 0 & 1-c_D \ceil*{\frac{n(k_A^R+1)}{2}}& 0\\
1 \& k_A^R=0 & n-1 & 0 & 0 & 1-(n-1) c_D & 0\\
2 & n-1 & 1 & 1 & 1- \tau_R-n c_D & \tau_R-c_A\\
4 & 0 & 0 & n-1 & 1-\tau-\tau_R-(n-1)c_D & \tau+\tau_R\\
\hline
\end{array}$$
}
\hfill{}
\caption{Properties of the different potential SPEs when $1-\tau-\tau_R > (n-1)c_D$ (Note: $k_A^R = \floor*{\frac{\tau_R}{c_A}}$).}
\label{tab:SPE}\vspace{-10mm}
\end{table}
\end{center}
\end{lemma}

\begin{proof}
First note that any connected network contains at least $n-1$ links. Conversely, any set of nodes can be made connected by using exactly $n-1$ links (any spanning tree is a solution). We consider a situation where $\mathbb{1}_{E_1\backslash E_A}=0$. Then, either $D$ decides not to heal the network and receives a utility of $U^* = \tau \mathbb{1}_{E_1} -c_D |{\mathcal E}_1|$, or it decides to heal it (by using at most $n-1$ links) and receives a utility of at least $\overline{U} = (1-\tau-\tau_R) +  \tau \mathbb{1}_{E_1} -c_D (|{\mathcal E}_1| + n-1)$. The difference is $\overline{U}-U^* =  (1-\tau-\tau_R) -c_D (n-1) > 0$. Thus, $D$ always prefers to heal the network after the attack of $A$. Therefore, Situations $3$ and $5$ contain no SPE. 

Next we consider Situation $4$. Since $\mathbb{1}_{E_1\backslash E_A\cup E_2}=1$, then $D$ needs to create in total at least $n-1$ links: $|{\mathcal E}_1|+|{\mathcal E}_2| \geq n-1$. Therefore, an optimal strategy is ${\mathcal E}_1 = \emptyset$ and $|{\mathcal E}_2| = n-1$. Since ${\mathcal E}_1 = \emptyset$, the optimal strategy of $A$ is ${\mathcal E}_A = \emptyset$.

In Situation $2$, $({\mathcal N}, {\mathcal E}_1)$ is connected, and thus $|{\mathcal E}_1| \geq n-1$. Further, $\mathbb{1}_{E_1}=1$ and $\mathbb{1}_{E_1\backslash E_A} = 0$, and thus $|{\mathcal E}_A| \geq 1$. 
Since $1-\tau-\tau_R >(n-1)c_D$, then $A$ should remove the minimum number of links to disconnect the network, and we obtain the result.

Finally, in Situation $1$, since $\mathbb{1}_{E_1\backslash E_A}=1$, then $D$ does not need to create any link during the healing phase: ${\mathcal E}_2= \emptyset$. Since $1-\tau-\tau_R > (n-1)c_D$, then $A$ attacks at most $k_A^R$ links if and only if it obtains a nonnegative reward, that is $k_A^R$ is the largest integer such that $\tau_R-c_Ak_A^R \geq 0$ which yields $k_A^R = \floor*{\frac{\tau_R}{c_A}}$. Thus, $D$ designs a network that is resistant to an attack compromising up to $k_A^R$ links. Such solution network is the ($|{\mathcal N}|,k_A^R+1$)-Harary network \cite{harary}.
\end{proof}


\subsection{Regime 2: $1-\tau-\tau_R < (n-1) c_D$}
\label{sec:2}

We now consider the case where $D$ has incentive, at phase $\tau+\tau_R$, to heal the network if at most $k$ links are required to reconnect it, where 
$ k = \floor*{\frac{1-\tau-\tau_R}{c_D}} < n-1.$

We study the potential SPE in Situations $3$, $4$ and $5$ in Lemma~\ref{lem:penible}, Situation $2$ in Lemma~\ref{lem:sit2}, and Situation $1$ in Lemma~\ref{lem:sit1}.

\begin{lemma}\label{lem:penible}
If $1-\tau-\tau_R < (n-1) c_D$, we have the following results:
\begin{itemize}
\item Any SPE in Situation $3$ satisfies ${\mathcal E}_2 = \emptyset$, $|{\mathcal E}_A| = k+1$ and $|{\mathcal E}_1| = n-1$,   leading to utilities $U_D = \tau - (n-1)c_D$ and $U_A = 1-\tau-(k+1)c_A$ (occurs only if $\floor{\frac{1-\tau}{c_A}}>k$);
\item There exists no SPE in Situation $4$;
\item The only potential SPE in Situation $5$ is the null strategy: ${\mathcal E}_1 = {\mathcal E}_2 = {\mathcal E}_A = \emptyset$, leading to utilities $U_D = 0$ and $U_A = 1$.
\end{itemize}
\end{lemma}

\begin{proof}
Suppose that an SPE occurs in Situation $5$. Since the network is always disconnected, then $U_D = - c_D (|{\mathcal E}_1|+|{\mathcal E}_2|)$. The maximum utility is obtained when ${\mathcal E}_1 = {\mathcal E}_2 = \emptyset$. Thus, ${\mathcal E}_A= \emptyset$.

In Situation $4$, since any connected network contains at least $n-1$ links, then the maximum utility of $D$ is  $U_D({\mathcal E}_1, {\mathcal E}_2, {\mathcal E}_A) =  (1-\tau-\tau_R) - c_D (n-1)<0$. Thus, $D$ is better off with a null strategy (occurring in Situation $5$).

In Situation $3$, since $\mathbb{1}_{E_1} = 1$ then $|{\mathcal E}_1| \geq n-1$. $D$ can achieves utility value $\tau - (n-1) c_D$ by playing a tree network. Since $\mathbb{1}_{E_1\backslash E_A} \neq \mathbb{1}_{E_1}$ then  $|{\mathcal E}_A| \geq 1$ and $U_A \leq 1-\tau - c_A$. The bound is achieved by attacking any one link created by $D$. We further can show that $A$ needs to attack $k+1$ links such that $D$ will not heal the network.
\end{proof}

In the following, we focus on the SPEs in Situations $1$ and $2$. In both cases, $\mathbb{1}_{E_1}=1$. Thus, $D$ creates initially a connected network. For each node $i \in {\mathcal N}$, let $d_i$ be its degree.  The potential best response strategies of $A$ to $E_1$ are summarized as follows:
\begin{enumerate}[label=(\roman*)]
\item Either $A$ does not attack and obtains a utility of $U_A^{(i)} = 0$;
\item Or $A$ attacks sufficiently many links so that the network admits $2$ components, i.e., $A$ attacks exactly $\min_{1 \leq i \leq n}d_i$ links to disconnect a node of minimal degree. Then, $D$ heals the network by constructing $1$ link, and $A$ receives utility 
\begin{equation}U_A^{(ii)} = \tau_R- (\min_{1 \leq i \leq n}d_i) c_A.\label{eq:uii}
\end{equation}
\item Or $A$ attacks sufficiently many links so that the network admits $\ell+2$ components, for some sufficiently large $\ell$ (whose exact value is discussed in the following two lemmas). Then, $D$ does not heal the network, and $A$ receives utility
\begin{equation}
U_A^{(iii)} = 1-\tau- |{\mathcal E}_A| c_A.\label{eq:uiii}
\end{equation}
 Note that any intermediate value of components in the range $\llbracket 2 ; \ell+2 \rrbracket$ cannot happen at SPE since it amounts to a lower utility for $A$.
\end{enumerate}

For convenience, We thus denote
$$ k_A^R = \floor*{\frac{\tau_R}{c_A}}\ \mathrm{and}\ k_A^H = \floor*{\frac{1-\tau}{c_A}}.$$
Note that $k_A^R$ (resp. $k_A^H$) corresponds to the maximal number of attacks that $A$ is willing to deploy to disconnect the network at phase $\tau_R$ (resp. $1-\tau$) so that $U_A^{(ii)}$ (resp. $U_A^{(iii)}$) achieves a positive value. With this in mind, we can show the following results:

\begin{lemma}\label{lem:sit2}
The only SPEs in Situation $2$ are such that $|{\mathcal E}_1|=n-1$, $|{\mathcal E}_A|=1$, $|{\mathcal E}_2|=1$,  $U_D = 1-\tau_R-n c_D$, and $U_A=\tau_R-c_A $. 
Furthermore, it occurs only if $c_A \leq \tau_R$ and 
$\floor*{\frac{1-\tau-\tau_R}{c_D}} > \floor*{\frac{1-\tau-\tau_R}{c_A}}.$
\end{lemma}

\begin{proof}
At an SPE in Situation $2$, the utility of $D$ is of the form $1-\tau_R - c_D (|{\mathcal E}_1| + |{\mathcal E}_2|)$. It is a best strategy if:
\begin{itemize}
\item It is the best strategy of $D$ to heal the network at time $\tau+\tau_R$, i.e., $1-\tau_R - (|{\mathcal E}_1| + |{\mathcal E}_2|) c_D \geq \tau -|{\mathcal E}_1| c_D.$
Thus, $|{\mathcal E}_2| \leq \floor*{\frac{1-\tau - \tau_R}{c_D}} := k$, and $k$ is the maximum number of links that $D$ can create at time $\tau+\tau_R$ at an SPE.
\item $D$ receives a better reward than by playing its best strategy in Situation $3$, i.e., $1-\tau_R - (|{\mathcal E}_1| + |{\mathcal E}_2|) c_D \geq \tau-(n-1) c_D.$
Thus, $|{\mathcal E}_1| + |{\mathcal E}_2| \leq \floor*{\frac{1-\tau - \tau_R}{c_D}} +(n-1).$ Note that $|{\mathcal E}_1| \geq n-1$. Since $k \leq n-1$, then altogether $D$ creates at most $|{\mathcal E}_1| + |{\mathcal E}_2| \leq 2(n-1)$ links. 
\end{itemize}

For any SPE in Situation $2$, we can write $|{\mathcal E}_1| = n-1+\alpha$ and $|{\mathcal E}_2| \leq k - \alpha$, for some $\alpha < k$. 
For Situation $2$, we obtain $U_A^{(ii)} \geq U_A^{(i)}$ which yields $(\min_{1 \leq i \leq n}d_i) \leq \floor*{\frac{\tau_R}{c_A}}$. If $\tau_R < c_A$, then no SPE exists in Situation $2$. 
Further, based on 
$0 \leq U_A^{(ii)} - U_A^{(iii)} = (|{\mathcal E}_A| - (\min_{1 \leq i \leq n}d_i))  c_A - (1-\tau-\tau_R)$, we obtain
$|{\mathcal E}_A| \geq \ceil*{\frac{1-\tau-\tau_R}{c_A}} + (\min_{1 \leq i \leq n}d_i)$.
Since at $\tau+\tau_R$, $D$ can create at most $k-\alpha$ links, then the goal of $A$ in case (iii) is to create at least $\ell = k-\alpha+2$ components in the network (that is, to create a $k-\alpha+1$ cut). Hence, $D$ constructs ${\mathcal E}_1$ in a way that at least $k_A + (\min_{1 \leq i \leq n}d_i)$ links need to be removed so that the network consists of $k+2-\alpha$ components. 

We denote $k:= \floor*{\frac{1-\tau - \tau_R}{c_D}}$ and $k_A:=\ceil*{\frac{1-\tau-\tau_R}{c_A}}$.
Suppose that $k < k_A$ (i.e. $k \leq k_A-1$). Then, for any $E_1$, consider the following attack: first remove $\alpha$ links so that the resulting network is a tree and then remove $k_2+1-\alpha$ links. Then, the resulting network has exactly $n-2-k+\alpha$ links, that is, it has $n-(n-2-k+\alpha) = k-\alpha+2$ components and is obtained using $k+1 < k_A + (\min_{1 \leq i \leq n}d_i)$ links. Thus, if $k < k_A$, no SPE in Situation $2$ exists.

If $k > k_A+1$ (i.e. $k \geq k_A$), then we consider the strategy that consists for $D$ to create a line network at time $0$. Then to create $k+2$ components, $A$ would need to remove $k+1$ links. However, due to $k > k_A+1$, it is not a best response for $A$. The best response for $A$ is to attack exactly one link (one being adjacent to one of the nodes with degree $1$). Then, the best strategy for $D$ is to recreate this compromised link at time $\tau+\tau_R$ which is an SPE. It is strategic as it minimizes the number of created links. All other SPEs, i.e., trees created at time $0$ and choices of the link to remove and the one to heal, yield the same payoffs for both players.
 \end{proof}

The following lemma characterizes the SPE in Situation 1:
\begin{lemma}\label{lem:sit1}
If $\tau_R/c_A > n-1$ or $\floor*{\frac{1-\tau}{c_A}} >  \floor*{\frac{1-\tau}{c_D}}$, then no SPE exists in Situation $1$. 
Otherwise, let (recall  $k= \floor*{\frac{1-\tau - \tau_R}{c_D}}$)
\begin{equation}
\delta = \left\{ 
\begin{array}{ll}
\ceil*{\frac{n \left( k_A^R+1\right)}{2}} & \text{if } k \geq 1 \text{ and } k_A^R  > 1, \\
\ceil*{\frac{n \left( k_A^H+1\right)}{2}} & \text{if } k = 0 \text{ and } k_A^R  > 1, \\
n & \text{if } k_A^H = k+1 \text{ and } k_A^R  = 1, \\
n+\floor*{\frac{n}{k}}+\ceil*{\frac{\floor*{\frac{n}{k}}}{2}} & \text{if } k_A^H \neq k+1 \text{ and } k_A^R  = 1, \\
n-1 & \text{if } k_A^H = k \text{ and } k_A^R  = 0, \\
n & \text{if } k_A^H \neq k \text{ and } k_A^R  = 0. \\

\end{array} \right. \label{eq:delta}
\end{equation}
If $1 < \delta c_D$ or if $1-\tau < (\delta-n+1) c_D$, no SPE in Situation $1$ exists.
Otherwise, the unique SPE is such that $U_D = 1 - \delta c_D$ and $U_A = 0$.
\end{lemma}


\begin{proof}
At an SPE in Situation $1$, $U_D$ is of the form  $1 - c_D |{\mathcal E}_1|$. Therefore, we obtain $ |{\mathcal E}_1| \leq \floor*{\frac{1}{c_D}}.$
Further, $1 - c_D |{\mathcal E}_1|$ should be greater than the utility in Situation $3$, i.e., $\tau-(n-1) c_D$. Thus, we obtain $1 -\tau  \geq (|{\mathcal E}_1|-(n-1)) c_D$. Since $1-\tau-\tau_R < (n-1) c_D$, then $1 -\tau  \geq (|{\mathcal E}_1|-(n-1)) c_D + 1-\tau-\tau_R - (n-1) c_D$, that is
\begin{equation}
\tau_R  \geq (|E_1|-2(n-1)) c_D.
\label{eq:4}
\end{equation}

The SPE in Situation 1 satisfies $U_A^{(i)} > U_A^{(ii)}$ and $U_A^{(i)} > U_A^{(iii)}$. 
Thus, the goal of $D$ is to create a network with the minimal cost such that all nodes have a degree of at least $\floor*{\frac{\tau_R}{c_A}}+1$, and at least $\floor*{\frac{1-\tau}{c_A}}+1$ links need to be removed to yield a network with $k+2$ components (i.e., the minimum $(k+1)$-cut requires at least $\floor*{\frac{1-\tau}{c_A}}+1$ links). 
If $k_A^R \geq 1$, we consider the strategy of $D$ that consists in creating an $({\mathcal N}, k_A^R+1)$ Harary network. Thus, 
\begin{equation}
|{\mathcal E}_1| \geq
\left\{
\begin{array}{ll}
\ceil*{\frac{n \left( k_A^R+1\right)}{2}} &\text{if}\  k_A^R \geq 2, \\  
n &\text{if}\ k_A^R = 1,\\
n-1 &\text{otherwise.}\\
\end{array}
\right. \label{eq:5}
\end{equation}

Denote $k_D^H:= \floor*{\frac{1-\tau}{c_D}}$.
Suppose that $k_D^H < k_A^H$. Suppose that at phase $0$, $D$ constructs a network with $(n-1)+\overline{k}$ links for some $\overline{k} \leq k_D^H$. Consider the strategy for $A$ that consists in attacking randomly $k_A^H$ links. Since $k_A^H \geq k_D^H \geq \overline{k}$, then the resulting network has less than $n-1$ links and is thus disconnected. At phase $\tau+\tau_R$, $D$ can reconstruct at most $ (n-1)+k^H_D-(n-1)-\overline{k} = k^H_D-\overline{k}$ links. Then, the network at phase $\tau+\tau_R$ would contain at most $(n-1)+\overline{k}-k_A^H+k^H_D-\overline{k} = (n-1)+k^H_D-k_A^H < n-1$ links, and the network is disconnected. Therefore,  no SPE exists in Situation $1$ if $k_R^A<4$ and $k_D^H < k_A^H$.

Conversely, suppose that $k_D^H \geq k_A^H$. We obtain
$$ 
\begin{array}{l}
k_A^H  \leq k_D^H\Rightarrow \frac{1-\tau}{c_A}-1 < \floor*{\frac{1-\tau}{c_A}} \leq 
 \floor*{\frac{1-\tau}{c_D}} \leq \frac{1-\tau}{c_D} 
\Rightarrow \\
\frac{1}{c_A}-\frac{1}{c_D} < \frac{1}{1-\tau} \Rightarrow
\frac{\tau_R}{c_A}-\frac{\tau_R}{c_D} < \frac{\tau_R}{1-\tau}<1 \Rightarrow
k_A^R \leq k_D^R.
\end{array}
$$
Similarly,
$$ 
\begin{array}{l}
k_A^H  \leq k_D^H\Rightarrow
\frac{1}{c_A}-\frac{1}{c_D} < \frac{1}{1-\tau} \Rightarrow\\
\frac{1-\tau-\tau_R}{c_A}-\frac{1-\tau-\tau_R}{c_D} < \frac{1-\tau-\tau_R}{1-\tau}<1 \Rightarrow
\floor*{\frac{1-\tau-\tau_R}{c_A}} \leq k.
\end{array}
$$

By definition, $k_A^H = \floor*{\frac{1-\tau}{c_A}} = \floor*{\frac{1-\tau-\tau_R}{c_A} + \frac{\tau_R}{c_A}} \leq \floor*{\frac{1-\tau-\tau_R}{c_A}} + \floor*{\frac{\tau_R}{c_A}} +1 \leq k + k_A^R+1$. Thus, $k_A^H \leq k + k_A^R +1$. 

\paragraph*{\textbf{Case \boldmath$k>0$}} If $k_A^R \geq 3$, then $k_A^R+1$ link removals are needed to disconnect the network, and any further additional component creation requires to remove at least $4/2 =2$ links. Thus, at least $2k+k_A^R+1$ link removals are necessary so that the network has $k+2$ components. Thus, $A$ does not attack the network. 
If $k_A^R=2$, and if $k \leq \floor*{\frac{n}{2}}$, then at least $k_A^R+1+2k$ link removals are required, and otherwise $k_A^R+1+\floor*{\frac{n}{2}}+(k-\floor*{\frac{n}{2}})$ link removals are necessary. Thus, $A$ does not attack the network.
\paragraph*{\textbf{Case \boldmath$k=0$}} In this case, $k_A^H \leq k_A^R +1$. $A$ only needs to disconnect the network since $D$ does not heal. Thus, $D$ creates an $({\mathcal N}, k_A^H+1)$ Harary network at phase 0.
Suppose that $k_A^R=0$. Then if $k_A^H = k$, then $D$ creates a tree network which is an optimal strategy. Otherwise, $k_A^H = k+1$ in which case $D$ creates a ring network.

Finally, suppose that $k_A^R=1$. If $k_A^H = k+1$, then the ring network, i.e., the $({\mathcal N}, 2)$-Harary network, is optimal for $D$. Otherwise, if $k_A^H = k+2$, then $D$ needs to create a network of minimal cost such that no $k$ cut exists with $k+1$ links. To this end, we consider the following network. For each $i \in {\mathcal N}$, we create a link between nodes $i$ and $(i+1) \mod n$ (ring network). Then, we connect node $k$ to node $2k$, and connect node $2k$ to node $3k$, and so on. If $\floor*{\frac{n}{k}}$ is even, then we connect node $k\floor*{\frac{n}{k}}$ to node $0$. Otherwise, we connect node $0$ to any node of the network excluding $1$ and $n-1$.
The resulting network contains no $k$ cut of size $k+1$ links and is minimal in terms of the number of links. The resulting utility for $D$ is $U_D = 1-(n+\floor*{\frac{n}{k}}+\ceil*{\frac{\floor*{\frac{n}{k}}}{2}})c_D$.
\end{proof}





The results of Lemmas \ref{lem:penible}, \ref{lem:sit2} and \ref{lem:sit1} are summarized in Table~\ref{tab:SPE2}.
\begin{center}
\begin{table}[h]
{\small
$$\hspace{-1em}\begin{array}{|@{}c@{\,}|@{}c@{\;}c@{\;}c@{\,}|@{\,}c@{\;\;}c@{}|}
\text{Situation} & |{\mathcal E}_1| & |{\mathcal E}_A| & |{\mathcal E}_2| & U_D & U_A \\
\hline
1 & \delta & 0 & 0 & 1-c_D \delta & 0\\
2 & n-1 & 1 & 1 & 1- \tau_R-n c_D & \tau_R-c_A\\
3 & n-1 & k+1 & 0 & \tau-(n-1)c_D & 1-\tau-(k+1)c_A\\
5 & 0 & 0 & 0 & 0 & 1\\
\hline
\end{array}$$
}
\hfill{}
\caption{Properties of the different potential SPEs when $1-\tau-\tau_R < (n-1)c_D$ (Note: $\delta$ is given by Eq.~\eqref{eq:delta}).}
\label{tab:SPE2}\vspace{-5mm}
\end{table}
\end{center}

\section{Case Studies}\label{case_study}
In this section, we use a case study of UAV-enabled communication networks to corroborate the obtained results. UAVs become an emerging technology to serve as communication relays, especially in disaster recovery scenarios in which the existing communication infrastructures are out of service \cite{tuna2014unmanned}. In the following, we consider a team of $n=10$ UAVs. The unitary costs of creating and compromising a communication link between UAVs for the operator/defender and adversary are $c_D = 1/20$ and $c_A = 1/8$, respectively. When the adversary attacks the network at phase $\tau = 0.3$, and the defender heals it after $\tau_R = 0.2$, the UAV-enabled communication network configuration at SPE is shown in Fig. \ref{UAV} which admits a tree structure, and $A$ does not attack the network at SPE. In addition, the utilities for $D$ and $A$ at SPE with $\tau_R\in[0, 0.6]$ are shown in Fig. \ref{utility_SPE}. The SPE encounters switching with different $\tau_R$. As $\tau_R$ increases, the UAV network operator needs to allocate more link resources to secure the network. Otherwise, the attacker has an incentive to compromise the communication links with a positive payoff. Specifically, when $\tau_R<0.375$, $A$ does not attack the UAV network, and $D$ obtains a positive utility by constructing a securely connected network. When $0.375<\tau_R<0.5$, the defender creates a connected network with the minimum effort, i.e., $9$ links, at phase $0$. In this interval, the attacker will successfully compromise the system during phase $[\tau,\tau+\tau_R]$, and the defender heals the network afterward. When $\tau_R$ exceeds $0.5$, the defender does not either protect or heal the network. The reason is that larger $\tau_R$ provides more incentives for the attacker to compromise the links and receive a better payoff. This also indicates that agile resilience (small $\tau_R$) is critical in mitigating cyber threats in the infrastructure networks.

\begin{figure}
\includegraphics[width=0.9\columnwidth]{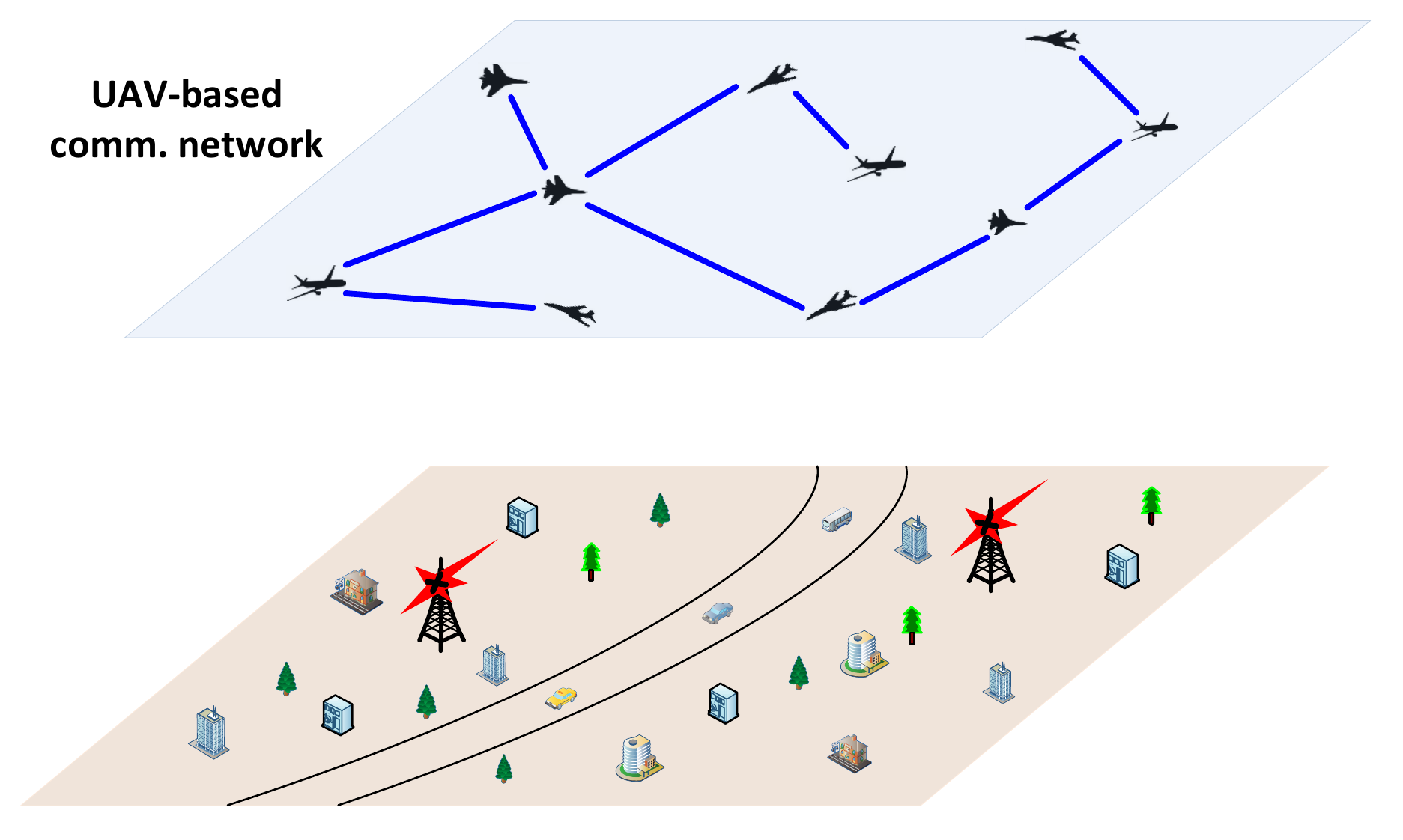}
\caption{UAV-based comm. network for disaster recovery. The UAVs form a tree network at SPE ($\tau=0.3$, $\tau_R=0.2$).}\label{UAV}
\end{figure} 

\begin{figure}
\includegraphics[width=0.9\columnwidth]{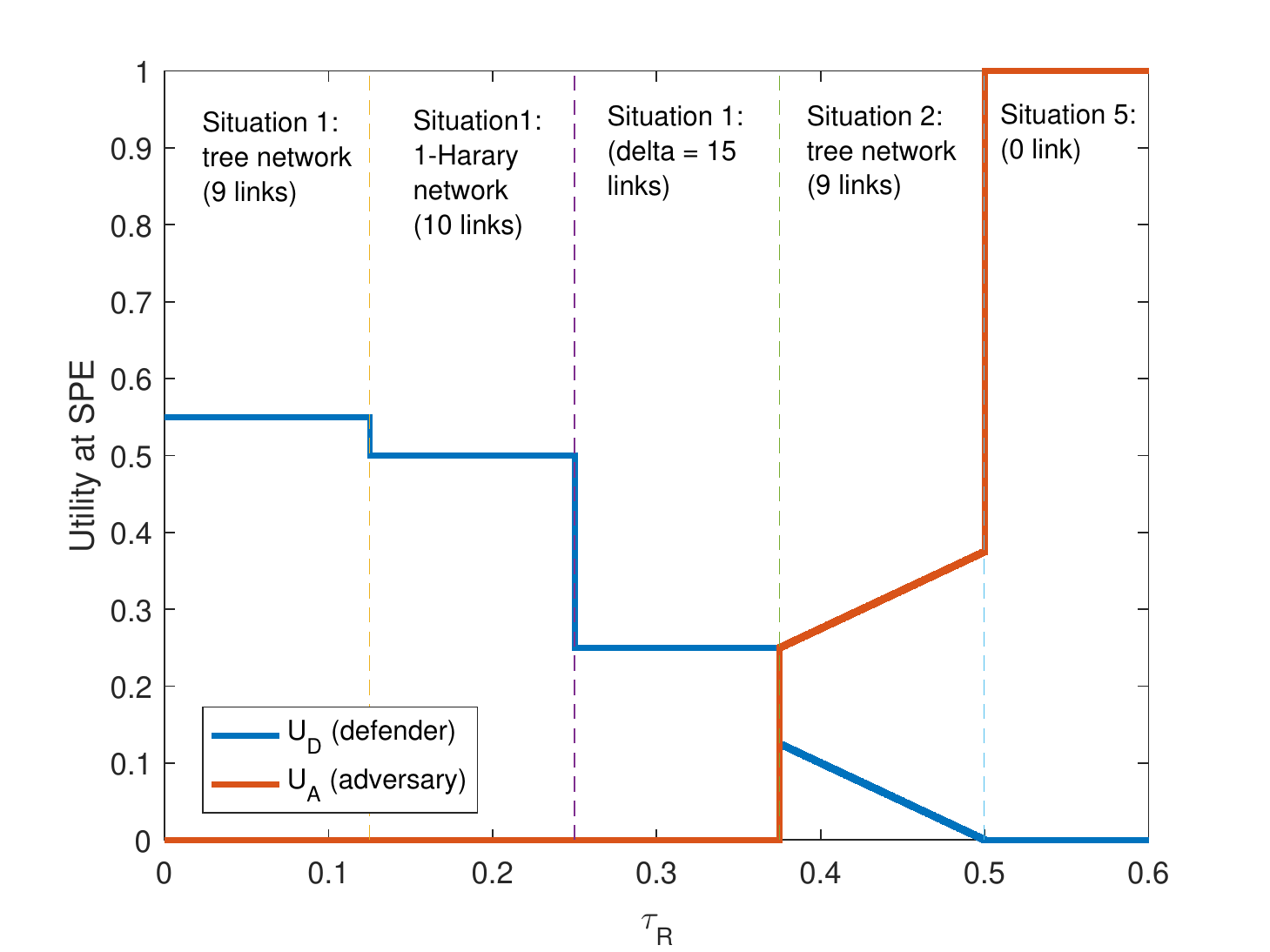}
\caption{Utilities for $D$ and $A$ at SPE with varying $\tau_R$.}\label{utility_SPE}
\end{figure} 

\section{Conclusion}\label{conclusion}
In this paper, we have established a two-player three-stage dynamic game for the infrastructure network protection and recovery. We have characterized the strategic strategies of the network defender and the attacker by analyzing the subgame perfect equilibrium (SPE) of the game. With a case study on UAV-enabled communication networks for disaster recovery, we have observed that with an agile response to the attack, the defender can obtain a positive utility by creating a securely connected infrastructure network. Furthermore, a higher level resilience saves link resources for the defender and yields a better payoff. In addition, a longer duration between the attack and recovery phases induces a higher level of cyber threats to the infrastructures.  Future work would include the extension of the network formation problem to interdependent networks and dynamic games with incomplete information.  


\bibliographystyle{ACM-Reference-Format}
\bibliography{ref} 

\end{document}